\documentclass[a4]{amsart}

\RequirePackage[utf8]{inputenc}
\RequirePackage[leqno]{mathtools}
\RequirePackage{amsthm}
\RequirePackage{amssymb}
\RequirePackage[mathcal]{euler}
\RequirePackage{dsfont}
\RequirePackage{varioref}
\RequirePackage[numbers]{natbib}
\RequirePackage[colorlinks=false,pdfborder={0 0 0}]{hyperref}
\RequirePackage{hypernat}
\RequirePackage{color}

\RequirePackage{enumerate}

\makeatletter
\@ifpackageloaded{hyperref}
{
	\def\mref#1#2{\hyperref[#2]{#1~\ref*{#2}}}
	\def\meqref#1#2{\hyperref[#2]{#1~(\ref*{#2})}}
	
} {
	\def\mref#1#2{#1~\ref{#2}}
	\def\meqref#1#2{#1~\eqref{#2}}
	
	\newcommand{\hyperref}[2][1]{#2}
}
\makeatother

\def\bigdissum{\sum}

\def\indep{\perp}
\def\alg{\mathcal}

\def\m{\mathfrak{m}}

\def\T#1#2{\mathcal{T}_{#1#2}} 

\def\symdiff{\vartriangle}
\def\tendsto{\longrightarrow}

\def\ssum_#1{\smashoperator{\sum_{#1}}}
\def\sbigdissum_#1{\smashoperator{\bigdissum_{#1}}}
\def\varkappa{\kappa}

\newcommand{\defem}[1]{\textbf{#1}}
\renewcommand{\emph}{\textit}
\newcommand{\dash}{\nobreakdash-\hspace{0pt}}

\DeclarePairedDelimiter\abs{\lvert}{\rvert}
\DeclarePairedDelimiter\norm{\lVert}{\rVert}
\newcommand{\clint}[2]{ \left[ #1, #2 \right] }
\newcommand{\opint}[2]{ \left( #1, #2 \right) }
\newcommand{\clopint}[2]{ \left[ #1, #2 \right) }

\def\bold{\mathds} 
\def\bld{\mathbf} 

\newcommand{\Z}{ \bold{Z} }
\renewcommand{\P}{ \bold{P} }
\newcommand{\Q}{ \bold{Q} }
\newcommand{\R}{ \bold{R} }

\newcommand{\tim}{\cdot}

\newtheorem{thm}{Theorem}
\newtheorem{rmk}{Remark}
\newtheorem{crr}{Corollary}
\newtheorem{con}{Conjecture}
\newtheorem{questions}{Questions}

\newtheorem{lem}{Lemma}
\newtheorem{df}{Definition}

\def\B{\mathcal{B}}
\def\bra#1{\left \langle #1 \right |}
\def\ket#1{\left | #1 \right \rangle}
\def\tr{\mathop{\mathrm{tr}}}
\def\proj{\mathop{\mathrm{Proj}}}
\def\bor{\mathop{\mathrm{Borel}}}
\def\mt{\widetilde \m}
\def\new#1{#1}

\begin{document}

	\begin{abstract}
		 We investigate the following generalisation of the entropy
		of quantum measurement. Let $H$ be an infinite-dimensional
		separable Hilbert space with a `density' operator $\rho$, $\tr \rho =1$.
		Let $I(\bold P)\in \R$ be defined for any partition
		$\bold P = (P_1,\ldots,P_m)$, $P_1+\ldots+P_m=1_H$, $P_i \in \proj H$
		and let $I(P_iQj, i\le m, j\le n) = I(\bold P) + I(\bold Q)$ for
		$\bold Q =(Q_1,\ldots, Q_n)$, $\sum Q_j = 1_H$ and
		$P_iQ_j  = Q_j P_i$, $\tr \rho P_iQ_j = \tr \rho P_i \tr \rho Q_j$
		($\bold P$, $\bold Q$ are physically independent). Assuming some
		continuity properties we give a general
		form of generalised information
		$I$, \mref{Theorem}{thm:main}, \meqref{formula}{eqn:thm:main}.
	\end{abstract}

	\title{On quantum information}
\author[A. Paszkiewicz]{Adam Paszkiewicz}
\email{ktpis@math.uni.lodz.pl}
		\address{
			Faculty of Mathematics and Computer Science\\
			University of Łódź\\
			ul. Banacha 22, 90-238 Łódź \\
			Poland\\}
\author[T. Sobieszek]{Tomasz Sobieszek}
\email{sobieszek@math.uni.lodz.pl}
\urladdr{http://sobieszek.co.cc}

\thanks{This paper is partially supported by Grant nr N N201 605840.}

\maketitle
	\section{Preliminaries and main results}	\label{sec:prelim}

Throughout the paper we shall use the following notations. $H$ is an infinite dimensional seperable
Hilbert space. By $\rho$ we denote a fixed positive trace-class operator with $\tr \rho = 1$. It is convenient
to identify $\rho$ with the functional $B(H)\ni A\mapsto \tr \rho A$ i.e. write $\rho(A)$
for $\tr\rho A$. We will also consider
\[
	\alg P = \{P\in \proj H : P,P^\perp\text{ are both infinite dimensional}\} \cup \{0,1_H\}
\]
We shall always write $\P$ for a sequence $(P_1,\ldots, P_m)$ of orthogonal
projections with $P_1 + \cdots + P_m = 1_H, P_i \in \alg P$. Thus $P_i$ are mutually orthogonal.
Every such $\P$ will be callled a \defem{partition} of~$1_H$. Let $\Q = (Q_1, \ldots, Q_n)$ be another
partition of~$1_H$. We write $\P \indep \Q$ when $\P$ and
$\Q$ are \defem{physically independent} i.e. when $P_iQ_j=Q_jP_i$ and
$\rho(P_iQ_j)=\rho(P_i)\rho(Q_j)$. In such case we shall write
$\P\tim\Q = (P_iQ_j;i=1,\ldots, m, j=1,\ldots n)$.
	\par
The paper is devoted to the investigation of the following general notion.

		\begin{df}
	We say that a real function $I$, defined on partitions of~$1_H$ is \new{an \defem{additive quantum  infrmation} (or an \defem{information} for short)} if  does not depend on the order of $P_1,\ldots, P_m$ and if
	\[
		I(\P\tim\Q) = I(\P) + I(\Q)
	\]
	for any $\P\indep \Q$. \end{df}


		\begin{df}		\label{df:Icontinous}
	Information $I$ is \defem{continous} if for any mutually commuting
	$P,P_1,P_2,\ldots \in \alg P$ such that \new{$0<\rho(P)<1$, and} $\rho\big(\abs{P-P_n}\big)\tendsto 0$ we have
	\[
		I(P_n,P_n^\perp) \tendsto I(P,P^\perp).
	\]\end{df}

	\begin{df}		\label{df:Ibounded}
	Information $I$ is \defem{bounded} if for any $0 < \varepsilon < 1$ the set of values
	\[
		\{I(P,P^\perp): P\in \alg P, \rho(P)=\varepsilon\}
	\]
	is bounded. \end{df}

		\begin{df}
	A \new{real} function $I_s$ on probability distributions that is sequences
	$\bld p = (p_1,\ldots, p_m)$, $p_i \ge 0$, $\sum p_i=1$ is called a \defem{symmetric information} if
	it does not depend on the order of elements $p_i$ and if for
	any probability distributions $\bld p = (p_1,\ldots, p_m)$, $\bld q = (q_1,\ldots, q_n)$ we have
	\[
		I_s( \bld p \otimes \bld q)  = I_s(\bld p) + I_s(\bld q).	
	\]
	By $\bld p \otimes \bld q$ we mean here a
	sequence $(p_i q_j: 1\le i \le m, 1\le j \le n)$.
	\end{df}

Given any function $I_s$ on probability distributions we shall write $(I_s\rho)(\P)$ for
$I_s(\rho(P_1), \ldots \rho(P_m))$.
	\par
Even though the notions of continuity and of boundedness of our information $I$ were introduced using
only $2$\dash element partitions $(P,P^\perp)$ they suffice to prove the following general
result.

		\begin{thm}	\label{thm:main}
	Let $I$ be any bounded, countinous information. There exists a self-adjoint trace-class
	operator $\mu$ with $\tr \mu = 0$ and a symmetric infromation $I_s$ such that
	\begin{equation}			\label{eqn:thm:main}
		I(\P) = (I_s\rho)(\P) + \ssum_{i=1}^m \tr \mu P_i \log \rho(P_i).
	\end{equation}
	\end{thm}

Here as throughout the paper we use base $2$ logarithms. The notation $(I_s \rho) (\bold P)$
denotes $I_s(\rho(P_1), \ldots, \rho(P_m))$. In the same way as for $\rho$
we shall denote $\tr\mu P$ by $\mu(P)$.
	\par
The proof is long and essentially depends on the following two non-trivial results. The first one
is the celebrated Gleason theorem. (c.f~\cite{V}, theorem 7.23). In our
notation the crucial part of this theorem can be formulated as follows

		\begin{thm}[Gleason]	\label{thm:Gleason}
	Let $p$ be any function $p:\proj H \mapsto \clint01$ satisfying
	$p\big(\sum_{k\ge 1} P_k\big) = \sum_{k\ge 1} p (P_k)$, for any sequence of mutually orthogonal projections
	$P_1,P_2, \ldots \in \proj H$, and $p(1_H)=1$. Then there exists a unique state $\rho$
	(i.e. a positive operator with $\tr \rho=1$)  satisfying $p(P) = \tr \rho P$. \end{thm}
	
The sum $\sum P_K$ above relates to strong (or equivalently weak) operator topology.
The next result is new and is contained in~\cite{PS}.

	\par
\new{With some exceptions, we will use the notation and  terminology introduced in~\cite{PS}. Nonetheless,
we present here all the denotations necessary for the statement of that result.} In particular,
we write $\bold A=(A_1,  \ldots, A_m),$ $\bold B = (B_1,\ldots, B_n)$
for any \new{finite} partitions of $\clopint01$ into borel  sets. Moreover we write
\[
	\bold A \indep \bold B \quad\text{if}\quad \lambda(A_i\cap B_j) =
		\lambda(A_i)\lambda(B_j),	\quad 1\le i \le m, 1\le j \le n,
\]
and
\[	\bold A \tim \bold B = (A_i\cap B_j: 1\le i \le m, 1 \le j \le n) \text{ for }
							\bold A \indep \bold B.
\]

		\begin{df}
	We say that \new{a real function $I_{\bor\clopint01}$ on finite partitions} is an \defem{information on a boolean structure}\footnotemark if
	\[
		I_{\bor\clopint01}(\bold A \tim \bold B) = I_{\bor\clopint01}(\bold A)  + I_{\bor\clopint01}(\bold B)
				\text{ for any }\bold A \indep \bold B.
	\] \end{df}
\footnotetext{\new{In paper~\cite{PS} an information on a boolean structure is called an \defem{additive partition entropy.}}}

		\begin{df}		\label{df:Icont}
	Information $I_{\bor\clopint01}$ on a boolean structure is \defem{continuous} if for any
	sequence $A, A_1, A_2, \ldots$ of borel subsets of $\clopint01$ \new{such that $0< \lambda(A)<1$},
	and $\lambda(A_n\symdiff A)\tendsto 0$ we have
	\[
		I_{\bor\clopint01}(A_n,A^c_n)\tendsto I_{\bor\clopint01}(A,A^c).
	\]	\end{df}

Despite the fact that continuity of $I_{\bor\clopint01}$ is defined using only $2$\dash element
partitions $(A, A^c)$, we have the following general result.

		\begin{thm}[\cite{PS}, Theorem \new{2}]	\label{thm:commutative}
	For any continuous information $I_{\bor\clopint01}$ on $\bor\clopint01$ there exists a unique
	signed measure $\m:\new{\bor\clopint01} \to \R$ and a symmetric information $I_s$ such that
	\begin{align*}
		&\m\big( \clopint01 \big) = 0, \\
		&I(\bold A) = (I_s\lambda)(\bold A) + \sum_{i=1}^m \m(A_i) \log\lambda(A_i)
	\end{align*}
	for any partition $\bold A = (A_1,\ldots, A_m)$. 
	Moreover $\m$ is absolutely continuous with respect to $\lambda$. 
	(The notation $(I_s\lambda)(\bold A)$  denotes
	$I_s(\lambda(A_1), \ldots, \lambda(A_m)).$)
	\end{thm}

The following definition makes it possible to transfer the above result into the
Hilbert space setting.

		\begin{df}
	By a \defem{boolean structure} we shall mean a lattice homomorphism
	$B:\bor\clopint01\mapsto \alg P$ such that
	\begin{align}
		B\Big(\bigcup A_j\Big) &= \sum B(A_j),	\quad\text{for disjoint }A_1, A_2, \ldots \in \bor\clopint01, \\
		\lambda(A) &= \rho(B(A)),\quad\text{for }A\in \bor\clopint01.
	\end{align}
		\end{df}

We shall denote the space of all boolean structures by $\B$.

The proof of \mref {Theorem} {thm:main} makes use of a certain connectedness property
of the family of all boolean structures $\B$. This property is shown in \mref {section} {sec:connectedness}.
The remaining part of the proof, which applies \mref {Theorems} {thm:Gleason}
and\mref{}{thm:commutative} is given is \mref{section} {sec:description}. Some remarks on
the assumptions of \mref{Theorem} {thm:main} and several natural conjectures are gathered
in \mref{sections} {sec:remarks}.1, \ref{sec:remarks}.2, \ref{sec:remarks}.3.
Given the proposed conjectures, it seems that the connectedness described in
\mref {section} {sec:connectedness}, \mref {Theorem} {thm:connected}, can have a few further applications.

\mref {Sections} {sec:remarks}.4 and \ref{sec:remarks}.5 gather some basic results and
show the role and interpretation \mref {Theorem} {thm:main} plays.

	\section{Connectedness of the space of boolean structures} \label{sec:connectedness}

In this section we show that we can pass from one boolean structure to another in small steps
by means of the following definition.

		\begin{df}
	Fix $k \ge 1$. We  say that $B, B_1 \in \B$ are \defem{$k$\dash equivalent}
	$(B \sim_k B_1 )$ if there are $B=B^0,\ldots, B^N=B_1 $ in $\B$ and sets
	$A_1,\ldots ,A_N\in \bor\clopint01$ such that for $1\le n\le N$ we have
	\[\begin{aligned}
		\lambda(A_n)	&\le  \tfrac1k,\\
		B^{n-1}(A)	&= B^n(A)	\quad\text{for }A\cap A_n=\emptyset .
	\end{aligned}\]
	We say that $B,B_1 $ are \defem{equivalent} $(B \sim B_1)$ if $B \sim_k B_1 $ for all $k\ge 1$.
	\end{df}

		\begin{thm} \label{thm:connected}
	For any $B,B_1 \in \B$ we have $B\sim B_1 $.
	\end{thm}

The proof is done by elementary reasoning and is devided into two steps:
	\par
1. For any boolean structures $B$, $B_1$ and $k\ge 1$ there exist boolean structures
$B'$ and $B'_1$ such that $B\sim_k B'$, $B_1 \sim_k B'_1$ and also $Q\in \alg P$ such that
\[\begin{aligned}
	&\rho Q = Q \rho, \quad \rho(Q) < \tfrac1k,\\
	&B'\clint0\varepsilon = B'_1\clint0\varepsilon=Q\quad\text{for some }
							\varepsilon \ge 0\text{ (\mref {Corollary} {crr:firststage}).}
\end{aligned}\]
	\par
2. Then for any $k\ge 1$ and any projection $Q$ with $\rho Q = Q \rho$,
$\rho(Q) < \frac1k$ we construct a partition $P_1 + \ldots + P_k = 1_H$ satisfying
$Q^\perp P_l Q^\perp = \frac1k Q^\perp$, $\rho(P_l) \le \frac2k$.
This is an example of dilatation method. This result is used to show that
$B' \sim_k B'_1$ for the boolean structures constructed in  step~1. (\mref {Lemma} {lem:thirdstage}.)
	\par
\mref {Theorem} {thm:connected} is a straightforward consequence of these  steps.

We begin with some auxillary properties of boolean structures.

		\begin{lem}\label{lem:Bexists}
	There exists $B\in \B$.  \end{lem}

		\begin{proof}
	Our state $\rho$ has a representation $\rho = \sum_{k\in \Z} \rho_k  \ket{e_k} \bra{e_k}$, where
	$\rho_k\ge 0$, in some orthonormal basis $(e_k)$ of the space~$H$.
	Consider the unitary operator $U:H \to L_2[0,1]$, given by
	\[
		(Ue_k)(x)	= e^{2\pi k i x},
	\]
	and the boolean structure $B$ where $B(A)=U^*1_A(\cdot)U$.
	Then
	\[
		\rho(B(A)) = \sum_{k\in \Z} \rho_k \norm{1_A(\cdot)e^{2\pi k i \cdot} }^2 = \lambda(A).
	\]
	\end{proof}

For such $B$ we obviously have $B(\{0\})=0$.

		\begin{lem}	\label{lem:BthroughPexists}
	Given $P_1 + \cdots + P_n = 1_H$, $P_i \in  \alg P$, $\rho(P_i)  > 0$, there exists $B\in \B$ with
	$B(\clopint {\alpha_{i-1}}{\alpha_i})=P_i$, for $\alpha_i = \rho(P_1 + \cdots + P_i)$, $0\le i\le n$.
	The same is true for a countable family of projections
	$\sum P_i = 1_H$, $\rho(P_i)>0$. \end{lem}

		\begin{proof}
	Set $\rho^i(\cdot)=\tfrac1{\alpha_i-\alpha_{i-1}}\rho(P_i\cdot P_i)$.
	Let  $B^i$ be a boolean structure in $P_i H, \rho^i$ in place of
	$H,\rho$ (\mref {Lemma} {lem:Bexists}). Now, it suffices to set
	\[
		B(A) = B^i \left(
			(A-\alpha_{i-1})\tfrac1{\alpha_i - \alpha_{i-1}} \right)
		\quad\text{for } A\in \clopint {\alpha_{i-1}} {\alpha_i}.
	\]
	\end{proof}

We shall take on the convention that $\clint \alpha\alpha = \{\alpha\}$ for $\alpha\in\R$.

		\begin{crr}	\label{crr:BthroughPexists}
	Given $P_1 + \cdots + P_n = 1_H$, $P_i \in  \alg P$, ($\rho(P_i)  = 0$ is possible now),
	there exists $B\in \B$ with $P_i \le B(\clint {\alpha_{i-1}}{\alpha_i})$, for
	$\alpha_i = \rho(P_1 + \cdots + P_i)$, $0\le i\le n$. \end{crr}

We also have

		\begin{lem}	\label{lem:BthroughD}
	Let $D\in\bor\clopint01$, with $\lambda(D)>0$. Given a projection $P\in\alg P$ with
	$\rho(P)=\lambda(D)$ there exists a projection-valued measure
	$B_D^P:\bor D\to \proj H$ such that
	\begin{equation}
		B_D^P(D)=P\quad\text{and}\quad\rho(B_D^P(A))=\lambda(A)\text{ for }A\in \bor D.
	\end{equation}	\end{lem}

		\begin{proof}
	Choose an arbitrary orthonormal basis $(e_k)$, $k\in\Z$ in $PH$ that satisfies
	$P\rho P = \sum_{k\in\Z} \rho^P_k \ket{e_k}\bra{e_k}$ for some $\rho^P_k \ge 0$.
	Now, consider a unitary operator $V:PH\to L_2(D)$ given by
	\[
		V(e_k)=\exp\left(	2\pi i k \tfrac{\lambda\left (D\cap\clopint0x \right) }{\lambda(D)}
					\right)
	\]
	It suffices to set $B_D^P(A)=V^*1_A(\cdot)V$ (c.f. \mref {Lemma} {lem:Bexists}). \end{proof}

	\begin{rmk}	\label{rmk:permute}
	For any boolean structure $B\in \B$, $k\ge 1$ and any permutation $\sigma$ of
	$\{1,\dots,k\}$ we have
	\[
		B\sim_k B^\sigma
	\]
	where
	\[
		B^\sigma(A)= B \left(	A-\tfrac l {2k} + \tfrac{\sigma(l)}{2k}	\right)
			\quad\text{for }
		A\subset\left[	\tfrac {l-1}{2k}\tfrac {l} {2k}	\right), \quad 1\le l \le 2k.
	\] \end{rmk}

Moreover:

		\begin{lem}	\label{lem:partition}
	For any partitions $A_1\cup \ldots \cup A_n= C_1\cup \ldots \cup C_n = \clopint01$, with
	$\lambda (A_l) = \lambda (C_l) \le \tfrac 1{2k}$, and $B\in \B$ there exists $B_1\in \B$ such that
	\begin{align*}
		B(A_l)	&= B_1(C_l)		\quad\text{for }	1\le l \le n,	\\
		B 		&\sim_k B_1
	\end{align*}
	\end{lem}
		\begin{proof}
	Consider a partition $E_1\cup \ldots \cup E_n$ which is independen both to $A_i$ and $C_i$
	(i.e. $\lambda(A_i\cap E_j) =\lambda(A_i)\lambda(E_j)$ and 
	$\lambda(C_i\cap E_j) =\lambda(C_i)\lambda(E_j)$). It suffices to prove the lemma
	with $E_i$ substituted for $C_i$.
		\par
	Consider a linear ordering
	\[
		(D_l)_{l=1}^{l=n(n-1)/2}	\text{ of the system } (A_i\cap E_j; 1\le i < j \le n),
	\]
	and denote $D'_l :=A_j \cap E_i$ for $D_l= A_i \cap E_j$.
		\par
	The sequence $B=B^0, B^1, \ldots, B^{n(n-1)/2}=B_1$ can be defined as follows
	\[
		B^{l+1}(A) = B^l(A)	\quad\text{for } A \cap (D_l \cup D'_l) = \emptyset,
	\]
	and
	\begin{align*}
		B^{l+1}(D_l) &= B^l(D'_l),	\\
		B^{l+1}(D'_l) &= B^l(D_l),
	\end{align*}
	which can be done in view of \mref {Lemma} {lem:BthroughD}.
	\end{proof}

We now come over to the first step in the proof of \mref {theorem} {thm:connected}.
The main objective  is  \mref {Corollary} {crr:firststage} below.
	\par
Given any nonzero vector $x\in H$ by $\hat{x}$ we shall denote
the projection $\tfrac {\ket x\bra x}{\norm x^2}$.

		\begin{lem}	\label{lem:bigvee}
	Suppose that $f_i\in H$, $\norm{f_i} \ge \varepsilon > 0$, $f_i \tendsto 0$ weakly.
	Then for any $\eta >0$ there exists $(g_i) \subset (f_i)$ with
	\[
		\rho\left(\bigvee _{j\ge k} \widehat{g_j}\right) \le \eta .
	\]
	\end{lem}

		\begin{proof}
	Observe that $\tfrac{f_i}{\norm{f_i}} \tendsto 0$ weakly. This gives $\widehat{f_i}\tendsto 0$
	weakly. For any finitely-dimensional projection $P$ we have
	$\norm{P^{\perp}f_i} \ge  \varepsilon/2$ for $i$ large enough and $P^\perp f_i \tendsto 0$
	weakly. Then $\rho \left( \widehat{P^\perp f_i} \right) \tendsto 0$, in particular
	\[
		\rho(P\vee \widehat {f_i}) = \rho (P) + \rho (\widehat{P^\perp f_i}) \tendsto \rho(P).
	\]
	By induction we can find a sequence $i(1) \le i(2) \le \ldots $ with
	\[
		\rho \left(	\bigvee_{1\le s \le t} \widehat{f_{i(s)}}	\right)
		\le \left(	1-\tfrac 1{2^k}	\right) \eta.
	\]
	Indeed, $\rho(P_{k+1} - P_k) = \rho\left( \widehat {P_k^\perp f_{i(k)}} \right) < \tfrac \eta {2^{k+1}}$ for
	$P_0 = 0$, $P_k= \widehat{f_{i(1)}} \vee\ldots \vee \widehat{f_{i(k-1)}}$ and for $i(k)$ large enough.
	\end{proof}

		\begin{lem}	\label{lem:firststage}
	For any orthonormal system $(e_n), n\ge 1$, $k\ge 1$ and a boolean structure $B\in\B$
	there exists a subsequence $(e_{n(i)})$ and $B' \in \B$ such that
	\[	\begin{aligned}
		B	&\sim_k	B', \\
		\widehat{e_{n(i)}} &\le B' \clopint0{\tfrac1k}.
	\end{aligned}	\]
	\end{lem}

		\begin{proof}
	Choose $1\le l \le 2k$ and a subsequence of indices $n_0(i)$ in such a way that
	\[
		\left\lVert
			B \clopint{
				\tfrac{l-1}{2k}
			}{
				\tfrac l{2k}
			}
			e_{n_0 (i)}
		\right\rVert^2 \ge \frac 1{2k} .
	\]
	There exists (c.f. \mref {Remark} {rmk:permute}) $B^0 \sim_k B$ that satisfies
	\begin{align}
			&B^0(A) = B(A)
				\qquad\text{for }
				A \cap	\left(
							\clopint 0{\tfrac 1{2k} }
							\cap
							\clopint {\tfrac{l-1}{2k}}	{\tfrac l{2k}}
						\right)
						= \emptyset,	\label{eqn:lem:firststage:Z}\\
			&B^0\clopint 0 {\tfrac 1{2k}} =
				B \clopint {\tfrac{l-1}{2k}}	{\tfrac l{2k}}, \notag \\
			&B^0\clopint {\tfrac{l-1}{2k}}	{\tfrac l{2k}} =
				B\clopint 0 {\tfrac 1{2k}}.	\notag 
	\end{align}
	In particular,
	\[
		\left\lVert
			B^0 \clopint{
				0
			}{
				\tfrac 1{2k}
			}
			e_{n_0 (i)}
		\right\rVert^2 \ge \frac 1{2k} .
	\]
	There exists a  second subsequence $n_1(i) \subset  n_0 (i)$ and a boolean structure
	$B^1\in \B$ satisfying
	\begin{align}
			& B^0(A) = B^1(A)
				\qquad\text{for }
				A\cap \clopint 0{\tfrac 1k} = \emptyset,	\label{eqn:lem:firststage:A}\\
			& \widehat {e_{n_1 (i)}}
				\perp
				B^1\clopint {\tfrac 1{2k}} {\tfrac 1k},	\label{eqn:lem:firststage:B}\\
			&\left\lVert
				B^1 \clopint{
					0
				}{
					\tfrac 1{2k}
				}
				e_{n_1 (i)}			\label{eqn:lem:firststage:C}
			\right\rVert^2 \ge \frac 1{2k},
			\qquad i \ge 1 .
	\end{align}
	In fact, $\norm{f_i} \ge \tfrac 1 {2k}$ for $f_i = B^0 \clopint 0 {\tfrac 1k} e_{n_0 (i)}$.
	Therefore there exists a subsequence $g_i = B^0 \clopint 0 {\tfrac 1k} e_{n_1 (i)}$ which
	satisfies $\rho(\bigvee_{i\ge 1} \widehat {g_i}) < \tfrac 1 {2k}$ (\mref {Lemma}{lem:bigvee}).
	Thus we can choose $B^1$ so that we would not only have
	\meqref {equation} {eqn:lem:firststage:A} but also
	\[\begin{aligned}
		&
			B^0 \clopint 0 {\tfrac 1k}
			=
			B^1 \clopint 0 {\tfrac 1k}, \\
		&
			B^1 \clopint {\tfrac 1 {2k}} {\tfrac 1k}
			\perp
			g_i,
	\end{aligned}\]
	that is \meqref {equation} {eqn:lem:firststage:B} and next\meqref{}{eqn:lem:firststage:C}.
		\par
	Continuing in this way we shall find sequences of indices
	$n_0(i) \supset  n_1 (i) \supset  \ldots \supset  n_{2k-1} (i)$ and structures
	$B^0, B^1, \ldots, B ^{2k-1}$ such that
	\begin{align}
			& B^{l-1}(A) = B^l(A)
				\qquad\text{for }
				A\cap \bigg( \clopint 0{\tfrac 1{2k}} \cap
				\clopint {\tfrac l{2k}} {\tfrac {l+1}{2k}} \bigg) = \emptyset,	\label{eqn:lem:firststage:D}\\
			& \widehat {e_{n_l (i)}}
				\perp
				B^1\clopint {\tfrac l{2k}} {\tfrac {l+1}{2k}},	\label{eqn:lem:firststage:E}\\
			&\left\lVert
				B^l \clopint{
					0
				}{
					\tfrac 1{2k}
				}
				e_{n_l (i)}
			\right\rVert^2 \ge \frac 1{2k},
			\qquad i \ge 1 ,\notag 
	\end{align}
	for $1 \le l \le 2k-1$. For $B' = B^{2k-1}$ we have
	$B \sim_k B'$ (by\meqref{}{eqn:lem:firststage:D} and\meqref{}{eqn:lem:firststage:Z}), and
	\[
		\widehat {e_{n_{2k-1} (i)}}
		\perp
		B_1\clopint {\tfrac 1{2k}} 1,
	\]
	by\meqref{}{eqn:lem:firststage:D} and\meqref{}{eqn:lem:firststage:E}.
	\end{proof}

Recall the convention that $\clint 0\alpha = \{0\}$ for $\alpha = 0$.

		\begin{crr}	\label{crr:firststage}
	For $B, B_1 \in \B$  and $k\ge 1$ there exist boolean structures
	$B'$ and $B'_1$ such that
	\[
		B\sim_k B' ,\quad B_1 \sim_k B'_1
	\]
	and also for some $0 \le \alpha \le \tfrac 1k$
	\begin{equation}	\label{eqn:crr:firststage}
		B'\clint0\alpha = B'_1\clint0\alpha=:Q,
	\end{equation}
	with $\dim Q = \infty $, $\rho Q = Q \rho$.
	\end{crr}

		\begin{proof}
	Let $\rho= \sum_{i \ge 1} \rho_n \widehat{e_n}$, where $(e_n)$ is a given
	orhonormal system, $\rho_n \ge 0$, $n \ge 1$. Using
	\mref {Lemma} {lem:firststage} we can find  subsequences $(e_n) \supset  (e_{n(i)}) \supset  (e_{m(i)})$
	and structures $B''$, $B''_1$ that satisfy
	\begin{align*}
		&B'' \sim_k B, \quad B''_1 \sim_k B_1, \\
		&
		\widehat {e_{n(i)}}
		\perp
		B''\opint {\tfrac 1{2k}} 1, \quad
		\widehat {e_{m(i)}}
		\perp
		B''_1\opint {\tfrac 1{2k}} 1,\quad i \ge 1.
	\end{align*}
	Then $Q= \sum \widehat { e_{m(i)} }$ satisfies
	\[
		Q \le B'' \clopint 0 {\tfrac 1{2k}} \wedge B''_1 \clopint 0 {\tfrac 1{2k}}
	\]
	i.e. \meqref{equation} {eqn:crr:firststage} for some
	$B' \sim_{2k} B''$, $B'_1 \sim_{2k} B''_1$.
	\end{proof}

We now come over to the second stage in the proof of \mref {Theorem} {thm:connected}.

		\begin{lem}	\label{lem:secondstage}
	A. Let $P,Q$, $P\perp Q$, be infinite-dimensional projections. For any $k\ge 1$ there
	exist mutually orthogonal projections $P_1, \ldots, P_k$ such that
	\begin{align}
		& P_1+\cdots + P_k = P+Q, 	\label{eqn:lem:secondstage:A} \\
		& PP_lP= \tfrac 1k P.			\label{eqn:lem:secondstage:B}
	\end{align}
		\par
	B. Whenever $\rho P = P\rho$ conditions \eqref{eqn:lem:secondstage:A}
	and \eqref{eqn:lem:secondstage:B} imply
	\begin{equation}
		\rho(P_l) \le \tfrac 1k \rho(P) + \rho(Q)
	\end{equation}
		\par
	C. For any partition $P=P^1 + \ldots + P^r$, conditions \eqref{eqn:lem:secondstage:A},
	\eqref{eqn:lem:secondstage:B} imply the existence of partitions
	\begin{align*}
		P_l &= P_l^1 + \ldots + P_l^r, \quad	0 \le l \le k, \\
		Q   &= Q^1 + \ldots + Q^r
	\end{align*}
	satisfying
	\begin{align}	\label{eqn:lem:secondstage:F}
		P^s + Q^s		&= \ssum_{1\le l \le k} P^s_l, \\
		P^sP_l^sP^s	&= \tfrac 1k P^s,	\notag
		\quad 1 \le s \le r,\ 1\le l \le k.
	\end{align}
	\end{lem}

		\begin{proof}
	For an arbitrary orthonormal system $(e_{l,n})^{l\le k}_{n\ge 1}$, consider projcections
	$P'_l = \sum_{n\ge1} \widehat{e_{l,n}}$, and
	$P' = \sum_{n\ge1} \widehat{(e_{1,n} + \cdots + e_{k,n})}$.
	Then $P'P'_lP' = \frac 1k P'$. We can find a partial isometry $U$ such that
	$UU^*= \sum_{1\le l \le k} P'_l$, $U^*U=P+Q$ and $U^*P'U=P$. This
	gives\meqref{}{eqn:lem:secondstage:A} and\meqref{}{eqn:lem:secondstage:B}
	for $P_l = U^*P'_lU$.
		\par
	If we also have $\rho P= P \rho$ then
	\begin{align*}
	\tr \rho P_l	&= \tr (P+P^\perp) \rho P_l  
				= \tr P \rho PP_l +  \tr P^\perp \rho P^\perp P_l  \\
				&= \tr \rho PP_lP + \tr  \rho P^\perp P_l P^\perp 
				= \tr \rho PP_lP  + \tr \rho Q P_l Q \\
				& \le \frac 1k \tr \rho P + \tr \rho Q.
	\end{align*}
		\par
	Moreover, for any partition $P = P^1 + \ldots + P^r$, the operator $U$ can be taken
	in such a way that 
	\[
		U^*P'^sU=P^s,	\quad 1\le s \le r,
	\]
	where $P'^s	= \sum_{n\ge 1} \widehat{ e^s_{1n} + \ldots + e^s_{kn} }$, for some grouping of the sequence
	$(e_{l,n})_{n\ge 1}$ into subsequences
	\[
		\big( e^1_{l,n} \big)_{n\ge 1}, \ldots, \big( e^r_{l,n} \big)_{n\ge 1}.
	\]
	It suffices to take
	\begin{align*}
		P'^s_l	&= \sum_{n\ge 1} \widehat{e^s_{ln} },	\\
		Q'^s &= \sum_{n \ge 1} \big( \widehat {e^s_{1n}} + \ldots + \widehat {e^s_{kn}} \big)
					- P'^s,
	\end{align*}
	and $P_l^s = U^*P_l'^s U$, $Q^s = U^*Q'^sU$, for $1\le s \le r$,
	$1\le l \le k$.  \end{proof}

		\begin{lem}	\label{lem:thirdstage}
	Consider a pair of boolean structures $B$, $B_1$ such that
	\[
		B\clint0\alpha = B_1\clint 0\alpha = Q,
	\]
	where $\dim Q = \infty $, $\alpha = \rho(Q) < \frac 1{4k}$, and $\rho Q = Q \rho$.  Then
	\[
		B\sim_k B_1.
	\]
	\end{lem}

		\begin{proof}
	Denote $Q^\perp = P = B\opint\alpha1$. Using \mref{Lemma} {lem:secondstage}
	we can find projections $P_1 + \cdots + P_{4k} = 1_H$, such that
	\begin{equation}	\label{eqn:PPlP}
		PP_lP = \tfrac 1{4k} P
	\end{equation}
	We shall focus our attention on $B$ for a while.
	Consider any partition into disjoint sets $A^1 \cup \ldots \cup A^{4k}= \opint \alpha 1$,
	with $0 < \lambda(A^s) \le \frac 1{4k}$.
		\par
	This partition generates (\mref {Lemma} {lem:secondstage} A., C.)
	\begin{align*}
		P		&= B(A^1) + \cdots + B(A^{4k}),	\\
		P_{l}		&= P^1_{l} + \cdots +  P^{4k}_{l},	\\
		Q		&= Q^1 + \cdots + Q^{4k},
	\end{align*}
	so that
	\[
		\ssum_{1\le l \le 4k} P^s_{l} = B(A^s) + Q^s, 
			\quad P^s_{l} = \tfrac 1{2k} P_{l} B(A^s) P_{l}
			\quad 1\le l,s \le 4k.
	\]
	We construct $B^0, B^1, \ldots, B^{2k} \in \B$ in the following way. Let
	\begin{align*}
		B^0(A) &= B(A)	\quad \text{for } A\cap \clint 0\alpha = \emptyset \\
		Q^s &\le B^0 \clint {\alpha^{s-1}}{\alpha^s}, \quad \text{for }
		\alpha^s = \rho(Q^1 + \cdots + Q^s), \text{ (then $\alpha^{4k}=\alpha$)},
	\end{align*}
	($\alpha^{s-1} = \alpha^s$ is possible, then
	$\clint {\alpha^{s-1}}{\alpha^s} =\{\alpha^s\}$ c.f. \mref{Corollary}{crr:BthroughPexists}).
	Obviously $B^0 \sim_{4k} B$. We now define
	\begin{align*}	
		B^s(A) &= B^{l-1}(A)\quad \text {for }
			A\cap \left( \clopint {\alpha^{s-1}}{\alpha^s} \cup A^s \right) = \emptyset
	\intertext{and}
		B^s(A^s_{l}) &= P^s_{l}
	\end{align*}
	for some partition $A^s_1 \cup \ldots \cup A^s_{2k} =
					A^s \cup \clopint{\alpha^{l-1}}{\alpha^s}$,
	(then the value of $B^{s}\{\alpha^s\}$ is also uniquely defined).
	Obviously $B\sim_{2k} B^{\alpha +1}$.
	Finally, as is easy to check
	\[
		B^{4k}(A^1_{l} \cup \ldots \cup A^{4k}_{l}) = P_{l}.
	\]
	We have obtained $B'=B^{4k}$ which satisfies
	\[
		B' \sim_{2k} B, \quad B'(A_{l}) = P_{l}
	\]
	for $A_{l} = A^1_{l} \cup \ldots \cup A^{4k}_{l}$.
		\par
	In a similar way we can build $B'_1$ such that
	\[
		B'_1 \sim_{2k} B_1, \quad B'_1(A_{1,l}) = P_{l}.
	\]
	for some partition $A_{1,1}\cup \ldots \cup A_{1,{4k}} = \clopint  01$.
	Moreover, we have
	\[
		\lambda(A_{l}) = \lambda(A_{1,l}) = \rho(P_{l})
	\]
	and $\rho(P_{l}) < \frac 1{2k}$ (by \eqref{eqn:PPlP}, the assumption
	$\rho(Q) < \tfrac 1{4k}$ and~\mref {Lemma} {lem:secondstage} B).
	\mref {Lemma} {lem:partition} leads to $B'\sim_k B'_1$. \end{proof}

The proof of \mref {Theorem} {thm:connected} follows directly from
\mref {Corollary} {crr:firststage} and \mref {Lemma} {lem:thirdstage}.

	\section{Description of quantum information} \label{sec:description}

Before we can combine  \mref {Theorem} {thm:connected} with
\mref {Theorem} {thm:commutative} which
describes the information on a single boolean structure we still need the following lemma.

		\begin{lem}	\label{lem:rational}
	Consider two continuous informations on boolean structures $I$ and $I_1$
	and let $\m$, $\m_1$ denote the measures of their corresponding nonsymmetric parts,
	(c.f. \mref {Theorem} {thm:commutative}).
	For $l\ge 1$ the condition
	\begin{equation}	\label{eqn:lem:rational:A}
		\clopint 0 {\tfrac 1l}\subset A_1
			\implies I(\bold A) = I_1(\bold A) ,
	\end{equation}
	for any partition $\bold A = (A_1, \ldots, A_n)$ of $\clopint 01$, implies
	\begin{equation}	\label{eqn:lem:rational:B}
		\m \left(\clopint 0 {\tfrac kl} \right) =
				\m_1\left(\clopint 0 {\tfrac kl} \right),	\quad\text{for all }1\le k \le l.
	\end{equation}
	\end{lem}

		\begin{proof}
	We can assume that $l \ge 3$. Let $\varepsilon \le \tfrac 1l$.
	Set
	\[
		\alpha = \sigma \left(\left\{\clopint {\tfrac il} {\tfrac {i+1} l}: \quad i= 0, \ldots, l-1 \right\}\right ),
	\]
	(the $\sigma${\dash}field generated by a partition).
		\par
	Observe that \meqref {condition} {eqn:lem:rational:A} implies
	\begin{equation}	\label{eqn:lem:rational:Aprimed}
		I(\bold A)= I_1(\bold A) \quad\text{for }\sigma(\bold A) \subset  \alpha
	\end{equation}
	Given arbitrary disjoint boolean sets $V,W$ and a partition $\bold A = (A_1, A_2,\ldots, A_n)$ such that
	$V \subset  A_1$, $W \subset  A_2$, by $\T VW \bold A$ we shall denote the partition
	\[
		(A_1\symdiff V \symdiff W, A_2\symdiff V \symdiff W,\ldots, A_n).
	\]
	Since $l \ge 3$, for any disjoint $V, W \in \alpha$ with $\mu(V) = \mu(W) = \tfrac 1l$
	there is a partition with $\sigma(\bold A) \subset  \alpha$, $V \subset  A_1$, $W \subset  A_2$ and such that
	$\lambda (A_2) = 2 \lambda (A_1)$.
	Then $\sigma(\T VW \bold A) \subset  \alpha$ and \mref {Theorem} {thm:commutative} gives
	\begin{align}
		I(\T VW \bold A) - I(\bold A)
			&=[\,\m(A_2\symdiff V \symdiff W)\log(2\lambda(A_1))  +  
				\m(A_1\symdiff V \symdiff W)\log(\lambda(A_1))\,]		\notag\\
			&\phantom{=} - [\,\m(A_2)\log(2\lambda(A_1)) +
				\m(A_1)\log(\lambda(A_1))\,] 							\notag\\
			&= [\,\m(A_2\symdiff V \symdiff W) - \m(A_2)\,] \log 2 			\label{eqn:IImm}\\
			&=\m(V)-\m(W).										\notag
	\end{align}
	In the same way
	\[
					I_1(\T VW \bold A) - I_1(\bold A)
					=\m_1(V)-\m_1(W).
	\]
	We have obtained that
	\[
					\m(V)-\m(W)=\m_1(V)-\m_1(W).
	\]
	This is satisfied also when $W=V$. Avaraging this equality over all considered $W$ and using
	$\m(\clopint 01) = \m_1(\clopint 01) = 0$ we get
	\[
		\m(V)=\m_1(V)
	\]
	i.e. \meqref {equation} {eqn:lem:rational:B}
	\end{proof}

From now on we fix a continous information $I: \bold P \mapsto I(\bold P) \in \R$.
Given a boolean structure $B \in \B$ by $\m_B$ we shall denote the measure
corresponding to the non-symmetric part of the continous 
information on a boolean structure $I \circ B:\bold A \mapsto (I \circ B)(\bold A)$, where
$(I \circ B)(A_1,\ldots, A_m) = I(B(A_1),\ldots, B(A_m))$.

		\begin{lem}		\label{lem:mBsame}
	For $B\clopint 0\alpha = P = B_1\clopint 0\alpha$ we have
	$\m_B(\clopint 0\alpha) = \m_{B_1}(\clopint 0\alpha)$. \end{lem}

		\begin{proof}
	Suppose first that $B(A) = B_1(A)$ for any  $A\in \bor \opint \alpha1$. Using
	\mref {Theorem} {thm:connected} and \mref{Lemma} {lem:rational} we can easly show
	that $\m_B =\m_{B_1}$ on $\clopint 0 {\frac l k}$ for any $\frac lk > \alpha$. By
	absolute continuity of $\m_B$, $\m_{B_1}$ with respect to $\lambda$ we get
	\[
		\m_B \clopint 0 \alpha  =  \m_{B_1} \clopint 0 \alpha
	\]
		\par
	The same conclusion will hold if $B(A) = B_1(A)$ for any  $A\in \bor \clopint 0\alpha$.
	For a general $B_1$ we only need to introduce $B_2$ so that
	\[
		B_2(A) = 	\left\{\begin{array}{ll}
				B(A);		\quad  &A\in\bor \clopint 0\alpha	\\
				B_1(A);	\quad  &A\in\bor \opint \alpha1
				\end{array}\right . .
	\]
	\end{proof}

Given last result we can define one
$\m:\alg P \to \R$ by setting
\[
	\m(P)=\m_B(P)
\]
whenever $B$ is a boolean structure with $B\clopint 0\alpha = P$
(\mref {Lemma} {lem:BthroughPexists}).

		\begin{lem}	\label{lem:misbounded}
	If a countinous information $I$ is bounded (see
	\mref {Definition} {df:Ibounded}) then the function $\m:\alg P \to \R$ that it
	generates is bounded and countably additive. \end{lem}

		\begin{proof}
	Suppose that for every partition $\bold P=(P,P^\perp)$  with $\rho(P)=1/3$ we
	have $\abs{I(\bold P)} \le M$. We will show that $\abs{\m(P)}\le 2M$ for any
	projection $P$, such that $\rho(P) \le 1/3$.
	Consider any $P \in \alg P$ as stated, and a boolean structure $B$ through $P$,
	i.e. $B(A) = P$, for some $A \in \bor \clopint 01$. Let $\m_B$ denote
	the measure corresponding to the non-symmetric part of $I_B$. Fix
	any set $V\in \bor \clopint 01$ with rational Lebesgue measure which is
	not greater that $1/3$. For any set $W\in \bor \clopint 01$ such that
	$\lambda(V) = \lambda(W)$ we can easily show that
	\[
		- 2M \le \m_B(V) -\m_B(W) \le 2M.
	\]
	In fact, take $V'=V\setminus W$, $W'=W\setminus V$, a partition
	$\bold A = (A_1,A_2)$, $V' \subset  A_1$, $W' \subset  A_2$ with
	$\lambda (A_2) = 2 \lambda (A_1) = 2/3$ and use a version of \eqref{eqn:IImm}.

	For some $l \ge 1$ we can avarage this inequality over all sets
	\[
		W \in \sigma(\{\clopint {i/l}{(i+1)/l}\}: 0 \le i < l\}), \quad\text{with }\lambda(W)=\lambda(V).
	\]
	Now using the fact that
	$\m_B(\clopint01) = 0$ we obtain
	$- 2M \le \m_B(V)  \le 2M$. Since $V$ was arbitrary
	with rational $\lambda(V) \le 1/3$ and since $\m_B$ is continous with respect
	to~$\lambda$ it follows that $- 4M \le \m_B(C)  \le 4M$ for any
	$C \in \bor \clopint 01$. This proves the boundedness
	of $\m$.
	
	To prove countable-additivity consider a partition
	$\sum_{n\ge 1} Q_n = Q$ of a projection $Q$ with $Q,Q_n \in \alg P$.
	Let $(P_n) = (Q_n; n\ge 1, \rho(Q_n) > 0)$ and $P= \sum P_n$.
	Then there exists a boolean structure~$B$ and a partition $A = \sum A_n$ into
	disjoint sets such that $B(A_n)=P_n$, $B(A) = P$
	(see proof of \mref {Lemma} {lem:BthroughPexists}). Then
	\[
		\m(Q) = \m(P) = \m_B(A) = \sum \m_B(A_n) = \sum \m(P_n) =\sum \m(Q_n)
	\]
	\end{proof}
	
In order to be able to use Gleason's theorem (in its classical form  given by \mref{Theorem} {thm:Gleason}) we need to extend $\m$ to the familly of all projections.

		\begin{lem}	\label{lem:mextension}
	Each function $\m:\alg P \to \R$ countably-additive on orthogonal projections
	has a unique extension to a countably-additive function $\mt:\proj H \to \R$.
	If $\m$ is bounded so is $\mt$.
	\end{lem}

		\begin{proof}
	Let $\alg S$ be the family of one-dimensional projections in~$H$. Given
	$e \in \alg S$, $P\in \alg P$ such that $e \perp P$ we write
	\[
		m^P_e  =  m(P+e) - m(P)
	\]
	We claim that $m^P_e$ does not depend on $P$. In fact, for $Q\in \alg P$, $e \perp Q$
	with $P \perp Q$, and $(P+Q)^\perp$ infinitely dimensional we have
	\[
		P+Q+e \in \alg P.
	\]
	Additivity of $\m$ gives
	\[
		\m(P+e)+\m(Q)=\m(Q+e) + \m(P)
	\]
	Thus
	\begin{equation}	\label{eqn:lem:mextension:A}
		\m^P_e = \m^Q_e
	\end{equation}
	For arbitrary $P,Q \in \alg P$, with $e \perp P$, $e \perp Q$ there exist $R, S \in \alg P$
	such that
	\begin{align*}
		&R\perp (P+e),\quad	S\perp(Q+e),\quad	R\perp S,	\\
		&(P+R)^\perp\!\!\!,\ (R+S)^\perp\!\!\!,\ (S+Q)^\perp
			\quad\text{are all infinite dimensional}
	\end{align*}
	Using\meqref{}{eqn:lem:mextension:A} we have
	\begin{equation}	\label{eqn:lem:mextension:B}
		\m^P_e = \m^R_e = \m^S_e = \m^Q_e
	\end{equation}
	Now, we can define
	\[
		\m_e = \m_e^P,
		\quad\text{for any }
		P\in \alg P, e \perp P .
	\]
	We show now that for $P \in \alg P$
	\begin{equation}	\label{eqn:lem:mextension:C}
		\m(P)=\sum \m_{e_i}\quad\text{if }P=\sum e_i .
	\end{equation}
	Indeed, take mutually orthogonal $Q_i \in \alg P$, with $P+\sum Q_i \in \alg P$.
	Then
	\begin{align*}
		\m(P)	&=\m(P+\sum Q_i) - m(\sum Q_i)	\\
				&=\sum \m(e_i+Q_i) - \sum m(Q_i)	
				=\sum \m_{e_i}.
	\end{align*}
	We are ready to define $\mt$,
	\[
		\mt(P) = \sum m_{e_i}
	\]
	if $P=\sum e_i$, $P \in \proj H$.
	To show that $\mt$ is well defined consider first finitely-dimensional projection
	$P=e_1+\cdots + e_n = f_1 + \cdots + f_n$, $e_i, f_i \in \alg S$.
	Take any $Q=\sum_{j\ge 1} g_j\in \alg P$ orthogonal to $P$, $g_j \in \alg S$.
	By\meqref{}{eqn:lem:mextension:C}
	\[
		\sum \m_{e_i} + \sum \m_{g_j} = \m(P+Q) = 
		\sum \m_{f_i} + \sum \m_{g_j}.
	\]
	We conclude showing that $\mt$ is well defined by considering $P\in \proj H$ such
	that $P^\perp = e_1 + \cdots e_n$ for some $e_i \in \alg S$. For any $(e_{n+i})_{i\ge 1}$ 
	such that $\sum_{i \ge 1} e_{n+i} = P$ we have
	\[
		\mt(P) = \sum_{i \ge 1} \m_{e_{n+i}} =
		\ssum_{i \ge 1} \m_{e_{i}} - \ssum_{1 \le i \le n} \m_{e_{i}}=
		0 - \mt (P^\perp).
	\]
	Countable additivity follows easily from definition of $\mt$. If $\m$ is bounded
	we need to show boundedness of $\mt$ on finitely dimensional projections.
	This follows from
	\[
		\mt(P) =  \m(P+Q) - \m(Q)
	\]
	whenever $P$ is finitely dimensional and $Q\in \alg P$ is orthogonal to it.
	\end{proof}
	
		\begin{lem}	\label{lem:Gleason}
	Given a bounded, countably-additive function $\mt:\proj H \to \R$ there exists a trace-class operator
	$\mu=\mu^*$, such that $\tr \mu =0$ and
	\[
	\mt (P) = \tr \mu P	\quad\text{for each }P\in\proj H.
	\]
	\end{lem}
	
		\begin{proof}
	For any space $K \subset H$, $3\le \dim K < \infty$, consider nonnegative additive functions
	$\proj K \ni P \mapsto \mt(P) + M\dim P$. Making use of Gleason's theorem we obtain
	an operator $\mu_K = \mu^*_K$ such that $\tr \mu_K Q = \mt(Q)$ for $Q \in \proj K$.
	The operator $\mu_K$ is uniquely defined by~$K$. In particular,
	\[
		\mu_K = P_K\mu_L P_K \quad \text {for } K \subset L \subset H,
	\]
	where $P_K$ is the orthogonal projection of $L$ onto~$K$. Thus
	\[
		\bra e \mu \ket f = \bra e \mu_K \ket f,	\quad\text{for }  K \ni e,f,
	\]
	is well defined. Then
	\[
		\mt \left(\sum \widehat {e_k}\right) = \sum \mt \big( \widehat{e_k} \big) 
			= \sum \bra {e_k} \mu \ket {e_k},
	\]
	for any orthonormal sequence $(e_k)$. This means that $\mu$ is a trace-class operator
	and $\tr \mu P = \mt (P)$ for any $P \in \proj H$.
	
	\end{proof}

The last three lemmas imply that given an information $I$ and a boolean structure $B$ there exists
a symmetric commutative information $I^B_s$ such that
\begin{equation}		\label{eqn:IB}
	I_s^B (\rho B(\bold A)) = I(B(\bold A)) - \sum_i \tr \mu B(A_i) \log \rho(B(A_i))
\end{equation}

		\begin{lem}	\label{lem:Isame}
	Given any $B, B_1$ we have
	\[
		I_s^B = I_s^{B_1}.
	\]
	\end{lem}

		\begin{proof}
	Fix $p_1,\ldots, p_n \ge 0$ with $\sum p_i=1$. We can assume that $p_1 >0$, and let $k > \tfrac 1{p_1}$.
	According to \mref {Theorem} {thm:connected} we have
	$B \sim_k B_1$. This means that there are $C_1, \ldots, C_M \in B\clopint 01$ and structures
	$B=B^0,\ldots,B^M=B_1$ 
	such that $B^{m-1} (A) = B^m(A)$ for $A\cap C_m = \emptyset$, $\lambda(C_m) \le 1/k$.
		\par
	Fix $1 \le m \le M$. Consider $\bold A = (A_1, \ldots, A_n)$ such that
	\[
		\lambda(A_i) = p_i	\quad\text{and}\quad	C_m \subset  A_1.
	\]
	then $B^{m-1} (\bold A) = B^m (\bold A)$.
	Finally,
	\begin{align*}
		I_s^{B^{m-1}}\!\big((p_i)\big)	= I\big(B^{m-1}(\bold A)\big) &-
			\sum_i \tr \mu B(A_i) \log p_i\\
					= I\big(B^{m}(\bold A)\big) &-
			\sum_i \tr \mu B(A_i) \log p_i
					= I_s^{B^m}\!\big((p_i)\big).
	\end{align*}
	\end{proof}

\mref {Theorem} {thm:main} is a direct consequence of \eqref{eqn:IB} and of
\mref {Lemma} {lem:Isame}.

	\section{Remarks and Conjectures} \label{sec:remarks}

	\subsection{Remarks on continuous quantum information}

We have shown that if information $I$ (defined on partitions $\bold P = \big( (P_1,\ldots, P_m)$:
$P_i \in \alg P \big)$ satisfies
\begin{enumerate}
	\item[($\alpha$)] is continous, in the sense of \mref {definition} {df:Icontinous},
	\item[($\beta$)] is bounded, in the sense of \mref {definition} {df:Ibounded},
\end{enumerate}
then it is of form~\eqref{eqn:thm:main}.

		\begin{rmk}
	Replacing the condition ($\beta$) with the simpler
	\begin{enumerate}
		\item[($\beta'$)] the set $\{I(P,P^\perp): P\in \alg P\}$ is bounded,
	\end{enumerate}
	would be too restrictive. \end{rmk}
	In fact, we will construct an information $I$ which satisfies ($\alpha$), ($\beta$) and does not satisfy~($\beta'$).
	Let $\rho$ be a state, which we now assume to be faithful (i.e. $\rho (A), A \ge 0$ implies $A=0$) and
	let $B$ be arbitrary boolean structure. Consider any
	\[
		\mu= - \widehat f_0 + \ssum_{i\ge 1}\tfrac 1 {2^i} \widehat f_i, \quad f_i \in B\clopint {2^{-2^{2i+2}}}{2^{-2^{2i}}},
		\quad \norm{f_i}= 1,
	\]
	and let
	\[
		I(\bold P) = \sum \mu(P_i) \log \rho(P_i)
	\]
	Since $\rho$ is faithful information $I$ satisfies condition ($\alpha$). While
	the function $P \mapsto \tr \mu P$ is bounded $I$ satisfies condition~($\beta$).
	However, for $P_n=B\clopint 0{2^{-2^{2n}}}$ we have
	\[
		\lim_{n\tendsto \infty} I(P_n,P_n^\perp) = \lim_{n\tendsto \infty}  \tr \mu P_n \tim \log \rho(P_n)=
		\lim_{n\tendsto \infty} {\tfrac 1{2^{n-1}}} \log 2^{-2^{2n}} = -\infty.
	\]

		\begin{rmk}
	The assumption ($\alpha$) is necessary. The boundedness ($\beta$) alone does not imply 
	\eqref{eqn:thm:main} of \mref{Theorem} {thm:main}. \end{rmk}
	
This is shown by the following example. Let $\xi$ be any nonnegative continous
functional on $l_\infty$, that satisfies $\xi((a_n)) = a$ whenever $\lim_{n\tendsto \infty}a_n=a$. Let $(e_n)$ be
an orthonormal system in $H$. The function $\m(P)=\rho(P)-\xi\big((\norm{Pe_n}^2)\big)$ for $P\in\alg P$
is finitely additive, however not countably additive on mutually orthogonal projections. Moreover $\m(1)=0$.
The function $I(\bold P)=\sum \m(P_i)\log \rho(P_i)$ satisfies the condition of boundedness ($\beta$) and
is not of shape~\eqref{eqn:thm:main}.

It is not obvious whether the condition $(\beta)$ is indispensible for getting~\eqref{eqn:thm:main}.
Before we pose other questions let us formulate a weaker condition of continuity.
\begin{enumerate}
	\item[($\alpha'$)] Whenever $P_1 \le P_2 \le \ldots  \in \alg P$, $P_n\tendsto P$, with $\rho(P)<1$ we have
		$I(P_n,P_n^\perp) \tendsto I(P,P^\perp)$.
\end{enumerate}

		\begin{questions}
	Is it true that for information $I$ the condition ($\alpha$) implies ($\beta$)?
	Does ($\alpha'$) imply ($\alpha$)? Do ($\alpha'$),($\beta$) imply ($\alpha$)?
	\end{questions}

	\subsection{Quantum information with no continuity assumptions}

The paper \cite{S} investigates informations on boolean structures $\bold A \to I(\bold A)\subset \R$ on
borel partitions $\bold A=(A_1,  \ldots, A_m)$ of the interval $\clopint 01$,
with no assumptions about continuity. Then we have the following
general result. (\cite {S}, Theorem~1)

Let $(\widehat \R,+)$ be the additive group of all endomorphisms of $(\R,+)$. Given
an information $I$ on a boolean structure there exists exactly one 'endomorphism-valued measure'
$A\mapsto \m(A)(\cdot) \in \widehat \R$ and exactly one symmetric information~$I_s$
on distributions $p_1+\ldots +p_n = 1$, such that

\[
	I(\bold A) = (I_s\lambda)(\bold A) + \sum_{i=1}^m \m(A_i) \big(\log\lambda(A_i)\big).
\]
By an endomorphism-valued measure we hereby mean a function satisfying
$\m(A)(\cdot) =\m(A_1)(\cdot)  + \ldots + \m(A_m)(\cdot)$
for $A_i\cap A_j = \emptyset$,$ i\neq j$, $\bigcup A_i = A$,

The following result which is analogical to \mref{Theorem} {thm:main}
can be easily obtained.

		\begin{thm}	\label{thm:mainonQ}
	For any information $I$ on partitions $\bold P = (P_1, \ldots, P_m)$ of $1_H$, with
	$P_i\in \alg P$ there exists a mapping
	$P \mapsto \m(P)(\cdot) \in \widehat \R$  defined on projections
	$P\in \alg P$, $\rho(P) \in \Q$ (the set of rationals) and a function $I_s$ defined on distributions
	$\pmb p = (p_1,\ldots,p_m)$, $p_i \in \bold Q$ such that
	\begin{equation}	\label{eqn:thm:mainonQA}
		\begin{gathered}
			\m(P_1 + \ldots +P_n) (\cdot) = \sum \m(P_i)(\cdot)\\
			I_s(\pmb p \tim \pmb q) = I_s(\pmb p) + I_s(\pmb q)
		\end{gathered}
	\end{equation}		
	for $\rho(P_i),p_i,q_j\in\bold Q$ and
	\begin{equation}
		I(\bold P) = \ssum_{1\le i \le n} \m(P_i) \big( \log \rho(P_i) \big)
					+ (I_s\rho)(\bold P)	\label{eqn:thm:mainonQB}
	\end{equation}
	for any partition $\bold P = (P_1,\ldots, P_n)$ with $\rho(P_i)\in \Q$.\end{thm}

		\begin{proof}
	An analogue of \mref {Lemma} {lem:rational} can be obtained for any (non-continous)
	informations $I$, $I_1$ and their endomorphism-valued measures $\m$, $\m_1$ on sets.
	Subsequently, \mref{Theorem} {thm:connected} can be used, just as in the proofs
	of \mref{Lemma}{lem:mBsame} and \mref{Lemma}{lem:Isame} to define the required
	endomorphism-valued measure $\m$ and the symmetric information~$I_s$.	\end{proof}

The following concjecture is much more interesting.

		\begin{con}
	In \mref {Theorem} {thm:mainonQ} the function $\m$ satisfying
	\eqref{eqn:thm:mainonQA} can  be defined for any $P \in \alg P$, $I_s$ can be definded for
	any distribution $\pmb p$, and
	\eqref{eqn:thm:mainonQB} is valid for any partition with $P_i\in \alg P$. \end{con}

	\subsection{Information on partitions with finitely-dimensional projections}

Let us consider information $I$ on the class of all partitions $\bold P = (P_1, \ldots,P_m)$
where $P_i \in \proj H$, that is we now allow $\dim P_i <\infty$. Again, we assume that $I(\bold P \tim \bold Q) = I(\bold P) + I(\bold Q)$ when the partions
$\bold P$, $\bold Q$ are physically independent (c.f. \mref{Section}{sec:prelim}). We say that
$I$ is \defem{continous} if
\begin{enumerate}
	\item[($\alpha$)] $I(P_n,P_n^\perp) \tendsto I(P,P^\perp)$ for any mutually commuting
	$P,P_1,P_2,\ldots \in \proj H$ such that $\rho\big(\abs{P-P_n}\big)\tendsto 0$.
\end{enumerate}
The information $I$ is \defem{bounded} if
\begin{enumerate}
	\item[($\beta$)]  for any $0 < \alpha < 1$ the set of values of
	$\{I(P,P^\perp): P\in \proj H, \rho(P)=\alpha\}$ is bounded.
\end{enumerate}

We have proved that \eqref{eqn:thm:main} is satisfied for $\bold P = (P_1,\ldots, P_m)$,
$P_i \in \alg P$. However  \eqref{eqn:thm:main} does not have to be satisfied
when some of the projections $P_i$ are finitely-dimensional.

	{\bf Example.} Fix one-dimensional projections $\widehat e \perp \widehat f$,
$\rho(\widehat e),\rho(\widehat f) > 0$. Denote by $\pi$ the class (of permutations)
of partitions
\[
	(\widehat e, \widehat f, P, P_1, \ldots, P_n) \subset \proj H; \quad \rho (P_i) = 0,
		\quad \text{for } 1 \le i \le n,
\]
and set $I(\bold P) = 1$ when $\bold P \in \pi$ and $I(\bold P)=0$ when $\bold P \not\in \pi$.

Then for $\bold P \in \pi$ and for $\bold Q$ being physically independent with $\bold P$ we
have $\bold P \tim \bold Q \in \pi$, $\bold Q \not \in \pi$. Moreover if the partitions $\bold P, \bold Q \in \pi$ then these partitions cannot be physically independent. Thus $I$ is an information,
moreover it satisfies ($\alpha$), ($\beta$). The formula \eqref{eqn:thm:main} is satisfied for
$\bold P = (P_1,\ldots,P_m)$, $P_i \in \alg P$ if and only if $I_s = 0$, $\mu = 0$. Then
\eqref{eqn:thm:main} is not satisfied for $\bold P \in \pi$.

We will give a condition,  stronger than  ($\alpha$), which makes such a
situation impossible.

		\begin{thm}
	Let $I$ be an information on partitions $\bold P = (P_1, \ldots, P_m)$, with $P_i \in \proj H$
	that satisfies the boundedness ($\beta$) and
	\begin{enumerate}
		\item[($\alpha$')] $I(\bold P^n) \tendsto I(\bold P)$ for any mutually commuting partitions
		$\bold P =(P_1,\ldots, P_m)$, $\bold P^n=(P_1^n,\ldots,P_m^n) \subset\proj H$ such that
		$\rho\big(\abs{P_i-P^n_i}\big)\tendsto 0$ as $n \tendsto \infty$.
	\end{enumerate}
	Then $I$ is of shape \eqref{eqn:thm:main} for any partitions $\bold P\subset \proj H$.
	\end{thm}

		\begin{proof}

	\new{For any partition of unity $\bold P =(P_1,\ldots,P_m)$, $P_i\in \proj H$, we can find partitions
	$\bold P^n = (P_i^n)$, $\bold Q = (Q_i)$, $\bold Q^n = (Q_i^n)$, 1$\le i \le m$, for $n\ge 1$ such that
	\begin{equation}\label{eqn:sA}\begin{aligned}
		P_i^n \nearrow P_i\quad&\text{or}\quad P_i^n \searrow P_i,\\
		Q_i^n \nearrow Q_i\quad&\text{or}\quad Q_i^n \searrow Q_i,
	\end{aligned}\end{equation}
	for $1 \le i \le m$, and
	\begin{align}
		&P_i^n,Q_i,Q_i^n\in \alg P,	\label{eqn:sB}\\
		&P_i-P_i^n = Q_i-Q_i^n, \quad\rho(P_i) = \rho(Q_i)	\label{eqn:sC}
	\end{align}
	for $1\le i \le m$, $n\ge 1$.}

	\new{In fact, this can be done as follows. For some $1\le j \le m$ we have $\dim P_i =\infty$. For the
	sake of simplicity let $j=1$. Then \[P_1=E + \ssum_{\substack{1\le i \le m,\\n\ge 0}} E_i^n\] for some
	infinite dimensional projections $E$, $E_i^n$, and
	\[
		E + \ssum_{2 \le i \le m} P_i = F + Q_2 + \ldots + Q_m
	\]
	for some projections $F, Q_2,\ldots, Q_m \in \alg P$ which satisfy
	$\rho(F) =\rho(E), \rho(Q_i) = \rho(P_i)$, $2 \le i \le m$, (as $\dim(E+\sum_{i \ge 2} P_i) =\infty$).}

	\new{It follows that
	\[
		Q_1 := 1_H - \ssum_{2\le i \le m} Q_i = F + \ssum_{\substack{2\le i \le m,\\n\ge 1}} E_i^n \in \alg P.
	\]
	Moreover we have $\dim P = \dim P^\perp = \infty$, and therefore $ P \in \alg P$,
	whenever $P$ is one of the projections
	\[\begin{aligned}
		P_1^n := P_1 -\ssum_{\substack{2\le j \le m, \\ k\ge n}} E_j^k,	\quad P_i^n := P_i + \ssum_{k\ge n} E_i^k,\\
		Q_1^n := Q_1 -\ssum_{\substack{2\le j \le m, \\ k\ge n}} E_j^k,	\quad Q_i^n := Q_i + \ssum_{k\ge n} E_i^k,
	\end{aligned}\]
	with $2\le i \le m$, and with $n\ge 1$. For just obtained partitions $\bold P = (P_i^n),$ $\bold Q = (Q_i)$, and
	$Q^n =(Q_i^n)$ all required conditions \eqref{eqn:sA}, \eqref{eqn:sB}, and \eqref{eqn:sC} are satisfied.}
	
	\new{Let us denote
	\[\begin{aligned}
		a_n &= (I_s\rho) \bold P^n + \ssum_{1\le i \le m} \mu(P_i^n) \log \rho(P_i^n) \\
		& -(I_s\rho) \bold  P- \ssum_{1\le i \le m} \mu(P_i) \log \rho(P_i) \\
		b_n &= (I_s\rho) \bold Q^n + \ssum_{1\le i \le m} \mu(Q_i^n) \log \rho(Q_i^n) \\
		& -(I_s\rho) \bold  Q- \ssum_{1\le i \le m} \mu(Q_i) \log \rho(Q_i).
	\end{aligned}\]
	then
	\[
		a_n-b_n = \sum\Big[\big(\mu(P_i^n)-\mu(Q_i^n)\big) \log \rho(P_i^n)
		- \big(\mu(P_i)-\mu(Q_i)\big) \log \rho(P_i)\Big], \text{ by~\eqref{eqn:sC}.}
	\]
	If for some $i$, $\rho(P_i) = 0$, then $\mu(P_i) =\mu(Q_i) = 0$ and $\mu(P_i^n) = \mu(Q_i^n)$, by~\eqref{eqn:sC},
	and we obviously assume that $0\cdot \infty = 0$. Thus $a_n - b_n$ tends to $0$, by~\eqref{eqn:sA}. On the
	other hand equations \eqref{eqn:sA}, and \eqref{eqn:sB} imply that
	$b_n = I(\bold Q^n) - I(\bold Q) \tendsto 0$, $a_n=I(\bold P^n) - (I_s \rho) \bold P - \sum \mu(P_i) \log \rho(P_i)$,
	and $I(\bold P^n) \tendsto I(\bold P)$, by~\eqref{eqn:sA}.} \end{proof}

	\subsection{A comparison of measures of information in classical case}

\new{It seems worthwhile to collect at the end some concepts of (additive) quantum information. We shall
do so in Section~4.5. First, however, we present some classical (commutative-probability) concepts of information as theories of increasingly general classes of functions~$I$.}

Throughout this section we shall assume that $\bold A \mapsto I(\bold A)$ is a function which is
\defem{additive} i.e. satisfies
$I(\bold A \tim \bold B) = I(\bold A)+I(\bold B)$ for measurable partitions $\bold A = (A_1, \ldots, A_m)$, 
$\bold B = (B_1,\ldots, B_n)$ of the interval $\clopint 01$ with $\bold A \indep \bold B$ i.e.
$\lambda(A_i \cap B_j) = \lambda(A_i)  \lambda(B_j)$.

\new{The classical results of Khinchin and Fadeev axiomatize the Shannon entropy with the use of the following minimal conditions}

		\begin{thm}[Rényi, Theorem 1, chapter IX. ]	\label{thm:Chinchin}
	Let $I$ satsify
	\begin{enumerate}[$1^\circ$]
		\item $I(\bold A) = I_s(\lambda(A_1),\ldots,\lambda(A_m))$ for
			(uniquely defined) function $I_s$ on finite probability distributions,
		\item $I_s(1/2,1/2)=1$,
		\item $I_s(p, 1-p)$ is a continuous function of $p$,
		\item $I_s(p_1,\ldots, p_m) = I_s(p_1+p_2,p_3,\ldots,p_m)
			+ (p_1 + p_2) I_s\Big(\tfrac{p_1}{p_1+p_2}, \tfrac{p_2}{p_1+p_2}\Big)$\footnotemark
	\end{enumerate}
	Then $I(\bold A) = \sum\lambda(A_i) \log \frac 1{\lambda(A_i)}$ (the Shannon entropy).
	\end{thm}
\footnotetext{\new{This version of the grouping axiom (c.f.~\cite{PS}) is often called
the~\defem{recursivity}. It implies the additivity of~$I$.}}

In a way, the idea of Rényi boils down to imposing less restrictive conditions on~$I$.
For any finite distribution $\pmb p =(p_1,\ldots, p_m)$ consider the
cumulative distribution function of Shannon entropy
\[
	F_{\pmb p}(x) = \sum_{\log p_i < x} p_i
\]

		\begin{thm}	\label{thm:Renyi}
	Let $I$ be additive and satisfy $1^\circ$ and~$2^\circ$ and
	\begin{enumerate}
		\item[$5^\circ$] $F_{\pmb p} \ge F_{\pmb q}$, $F_{\pmb p} \neq F_{\pmb q}$ imply
			$I_s(\pmb p) < I_s(\pmb q)$,
		\item[$6^\circ$] if the distributions $\pmb p, \pmb p^1, \pmb p^2, \pmb p^1_t, \pmb p^2_t$
			satisfy for $0 \le t \le 1$
			\begin{align*}
				I_s(\pmb p^1) &= I_s(\pmb p^2), \\
				F_{\pmb p_t^\epsilon}(x) &=
					tF_{\pmb p}(x) + (1-t)F_{\pmb p^\epsilon}(x), \quad \epsilon = 1,2,
			\end{align*}
			then $I_s(\pmb p_t^1) = I_s(\pmb p^2_t)$,
		\item[$7^\circ$] for $\epsilon > 0$, $M>0$ there exists $\delta > 0$ such that
			$\abs{I_s(\pmb p) - I_s(\pmb q)} < \epsilon$ when
			$\abs{F_{\pmb p} - F_{\pmb q}} < \delta$ for $x \in \R$, and $\pmb p, \pmb q$
			are concentrated on an interval of length~$M$.
	\end{enumerate}
	Then $I=I_\alpha$ for some $\alpha \in \R$ when
	\begin{align}
		I_1(\bold A) &= \sum\lambda(A_i) \log \frac 1{\lambda(A_i)}; \\
		I_\alpha(\bold A) &= \tfrac 1{\alpha -1} \log \sum \lambda(A_i)^\alpha
			\qquad\text{for }\alpha \neq 1.
	\end{align}
	\end{thm}

Outline of the proof will be given somewhat later. The classical interpretation
of informations $I_\alpha$ is provided only for the special case of $\alpha >0$
(see \cite{R}, chapter IX, $7$) as was stressed by A. Rényi.

The paper \cite{PS} considers additive functions~$I$ under a weak assumption on continuity
(\mref{definition} {df:Icont}).
Moreover, all the assumptions $1^\circ$–$7^\circ$ are dropped.
\mref{Theorem} {thm:main}, we cited in \mref{Section}{sec:prelim}, says that
$I(\bold A) = (I_s\lambda) (\bold A) + \sum \m(A_i)\log\frac1{\lambda(A_i)}$.

In order to obtain the simplest interpretation of this formula let us confine ourselves to the special case of
\begin{equation}	\label{eqn:conditionalI}
	I(\bold A) = \sum p_E (A_i) \log \frac1{\lambda(A_i)}
\end{equation}
for $E \subset \clopint 01$ and for conditional probability $p_E(A) = \lambda(A\cap E)/\lambda(E)$.
Assume now that the outcome of the experiment $A_i$ always `carries information' of
weight $\log \lambda(A_i)$, in accordance with the basic interpretation of Shannon entropy
(\mref {Theorem} {thm:Chinchin}). Then formula \eqref{eqn:conditionalI} gives the
conditional expectation of information carried by the experiment~$\bold A$, under condition $E$.

\mref{Theorem} {thm:Renyi}, as given here, requires a bit of explanation. A.~Rényi was
seeking a description of the gain of information between two distributions. As such
he was solving a somewhat different problem. (cf Theorem IX.6.1 in \cite{R}).
\mref{Theorem} {thm:Renyi} however is a \new{relatively} simple consequence of Rényi's fundamental
theorem on a functional of cumulative distribution functions (analysis of Postulates I',III',V,VI Chapter IX.6 in \cite{R}).

		\begin{thm}	\label{thm:RenyiFundamental}
	Let $J(F) \in \R$ be a number defined for each cumulative distribution function
	of a finite distribution and let the following  conditions be satisfied
	\begin{enumerate}[i)]
		\item $J(D_1)=1$ for $D_1$ being the cumulative of~$\delta_1$,
		\item $J(F * F_1) = J(F) + J(F_1)$,
		\item $F \le F_1$, $F \not \equiv F_1$ implies $J(F) > F(F_1)$,
		\item $J(F_1) = J(F_2)$ implies $J(tF + (1-t)F_1) = J(tF +(1-t)F_2)$ for $0 \le t \le 1$,
		for arbitrary cumulative distribution functions $F,F_1,F_2$ of finite distributions. Then
		$J=J_\alpha$, $\alpha \in \R$, where
		\begin{align*}
			J_1(F) &= \int x F(dx),\\
			J_\alpha(F) &= \tfrac 1 {\alpha -1} \log \int 2^{\alpha x} F(dx)\quad\text{for }\alpha\neq1.
		\end{align*}
	\end{enumerate}	\end{thm}

\subsection{The comparison of measures of information in the quantum case}

A good description of quantum measurements is given by physical indepence of
partitions $1_H$ onto mutually orthogonal projections (see \mref{Section}{sec:prelim}).

For this reason it is natural to formulate the conditions imposed on information $I(\bold P)$
on partitions $(\bold P) = (P_1,\ldots,P_m)$ of unit $1_H$ by using `cuts of $I$ to boolean structures'.
One needs for instance to assume that $I_B(\cdot) = (I\circ B) (\cdot)$ for boolean structures
$B:\bor \clopint01 \to \alg P$, $\rho(B(A)) = \lambda(A)$.

We will always assume that $I_B(\bold A \circ \bold B) = I_B(\bold A) + I_B(\bold B)$ for
partitions $\bold A \perp \bold B$ of the interval $\clopint 01$ and for any boolean structure~$B$.

A more limiting additional assumtion on functional $I_B$ by Fadeev and Rényi automatically give:

		\begin{thm}	\label{thm:vonNeumann}
	If $I_B$ satisfies also the conditions $1^\circ$, $2^\circ$, $3^\circ$, $4^\circ$, then $I=I_1$ where
	\begin{equation}	\label{eqn:thm:vonNeumann}
		I_1(\bold P) = \sum\rho(P_i) \log \frac 1{\rho(P_i)}\quad \text{for }\bold P = (P_1,\ldots,P_m), P_i \in \alg P,
	\end{equation}
	is a von Neumann's information.
\end{thm}

		\begin{thm}
	If $I_B$ satisfies also the conditions $1^\circ$, $2^\circ$, $5^\circ$, $6^\circ$, $7^\circ$ then $I=I_\alpha$ where
	$I_1$ is von Neumann's information \eqref{eqn:thm:vonNeumann}, while
	\begin{equation}
		I_\alpha(\bold P) = \tfrac 1{\alpha -1} \log \sum\rho(P_i)^\alpha
			\quad\text{for }\alpha \neq 1.
	\end{equation}	\end{thm}
		\begin{proof}
	According to \mref{Theorem}{thm:Renyi}, there is a number $\alpha(B)$ with $I_B(\bold A) = I_{\alpha(B)}(\bold A)$. Then $\alpha(B)$ is determined by the values of $I_B(\bold A)$ for $\bold A  = (A_1,\ldots, A_m)$, with
	$A_i \cap \clopint 0 {\tfrac 12} = \emptyset$. By \mref{Theorem}{thm:connected} there exists
	$\alpha=\alpha(B)$ independent from~$B$. \end{proof}

	A major difficulty crops up when we take on only weak assumptions on the continuity of the function~$I_B$.
	Our \mref{Theorem}{thm:main} gives an (almost) exhaustive answer.
	
	A particular case of our formula \eqref{eqn:thm:main} is
	\begin{equation}	\label{eqn:thm:main:particular}
		I(\bold P) = \sum \frac {\rho (E P_i E)}{\rho(E)} \log \frac1{\rho (P_i)}
	\end{equation}
	for a fixed projection $E \in \alg P$, $\rho(E) > 0$. Let us suppose that the measure of information
	contained in the outcome~$P_i$ of an experiment described by $\bold P$ is given by the number
	$\log \rho(P_i)$. 
	Then the quantity \eqref{eqn:thm:main:particular} can be interpreted as a conditional avarage information
	when we know that the event $E$ has occured and we avarage with respect to the state
	$P \mapsto \rho(EPE)/\rho(E)$.
	
	It should be explained that \meqref {formula} {eqn:thm:vonNeumann} gives
	the entropy of quantum measurement in von Neumann's sense, see~\cite{N}. Such a measurement
	is described by the partition~$\bold P$.
	In the simplest case, when the initial state $\rho$ is simple i.e. $\rho =\widehat e$ or $\rho(\cdot) =
	\langle \cdot e | e\rangle$, the state after measurement is given by
	$\rho_{\bold P} = \sum \rho(P_i)\widehat {e_i}$,  for $\widehat {e_i}  = P_i e / \norm{P_i e}$.
	Then $I_1(\bold P)$ is the famous von Neumann \emph{entropy of the state} $\rho_{\bold P}$, (and
	thus it is an information given by measurement~$\bold P$).
	This quantity, which was introduced by  von Neumann in \cite{N}, was widely
	investigated what can be found in \cite{OP}, \cite{P}.


\section*{References}
\begin{thebibliography}{2}
	\bibitem{N}
		\textsc{von Neumann, J.} (1954)
		\textit{Mathematical Foundations of Quantum Mechanics}
		Dover
	\bibitem{OP}
		\textsc{Ohya, M., Petz, D.} (1993)
		\textit{Quantum Entropy and Its Use},
		Springer
	\bibitem{P}
		\textsc{Petz, D.} (2008)
		\textit{Quantum Information Theory an Quantum Statistics}
		Springer
	\bibitem{R}
		\textsc{Rényi, A.} (1970).
		\textit{Probability Theory},
		Akadémiai Kiadó
	\new{\bibitem{PS}
		\textsc{Paszkiewicz, A., Sobieszek, T.}
		\textit{Additive entropies of partitions}.
		(preprint, {\tt arXiv:1202.4591})
	\bibitem{S}
		\textsc{Sobieszek, T.}
		\textit{Noncontinuous additive entropies of partitions}.
		(preprint, {\tt arXiv:1202.4590})}
	\bibitem{V}
		\textsc{Varadarajan, V. F.} (1968).
		\textit{Geometry of Quantum Theory}.
		D. Van Nostrand Company, Inc.
\end{thebibliography}
\end{document}